\documentclass[11pt]{amsart}

\usepackage[hmargin=1.2in,vmargin=1.2in]{geometry}
\usepackage{amsmath,amsthm,amssymb,amsfonts,verbatim}
\usepackage{hyperref}

 \providecommand{\F}{\mathbb{F}}

\title[LCD MDS codes]{Algebraic geometry codes with complementary duals exceed the asymptotic Gilbert-Varshamov bound}

\author[L. Jin ]{Lingfei Jin and Chaoping Xing}

\thanks{Lingfei Jin is with School of Computer Science, Shanghai Key Laboratory of Intelligent Information Processing,  Fudan University, Shanghai 200433, China. She is also with State Key Laboratory of Information Security (Institute of Information Engineering, Chinese Academy of Sciences, Beijing 100093). {\it email:} lfjin@fudan.edu.cn.}  
\thanks{Chaoping Xing is with School of Physical and Mathematical Sciences, Nanyang Technological University. {\it email:} xingcp@ntu.edu.sg.}

\newtheorem{lemma}{Lemma}[section]
\newtheorem{theorem}[lemma]{Theorem}

\newtheorem{cor}[lemma]{Corollary}

\theoremstyle{remark}
\newtheorem{rmk}{Remark}

\renewcommand{\epsilon}{\varepsilon}
\renewcommand{\le}{\leqslant}
\renewcommand{\ge}{\geqslant}

\def\PP{\mathbb{P}}

\def \mL {\mathcal{L}}

\def\wt{{\rm wt}}

\def\Tr{{\rm Tr}}

\newcommand{\Gg}{\gamma}     
\newcommand{\Gd}{\delta}     
\newcommand{\Ge}{\epsilon}

\newcommand{\Gl}{\lambda}

\def \ba {{\bf a}}
\def \bb {{\bf b}}
\def \bc {{\bf c}}

\def \bu {{\bf u}}
\def \bv {{\bf v}}
\def \bo {{\bf 0}}

\def\supp {{\rm supp }}

\def\res {{\rm res}}

\begin{document}

\maketitle

\begin{abstract}
It was shown by Massey that linear complementary dual (LCD for short) codes are asymptotically good. In 2004, Sendrier proved that LCD codes meet the asymptotic Gilbert-Varshamov (GV for short) bound. Until now, the GV bound still remains to be the best asymptotical lower bound for LCD codes. In this paper, we show that an algebraic geometry code over a finite field of even characteristic is equivalent to an LCD code  and consequently there exists a family of LCD codes that are equivalent to algebraic geometry codes and  exceed the asymptotical GV bound.

\end{abstract}

\section{Introduction}
Due to applications in  communication system, storage system, and cryptography, linear complementary dual codes have received much attention since they were introduced. An  LCD code refers to a linear code which has trivial intersection with its dual code. LCD codes have been extensively studied and many results and properties on LCD codes were given  \cite{DKOSS,EY09,ML86,Jin01,TH70,Mass92,YM94,Kim01}.  Some interesting results have been obtained in the literature.  Most of constructions  are based on cyclic codes such as BCH codes \cite{EY09,Mass64,LDL,TH70,YM94}. An optimal family of LCD codes, i.e.,  LCD MDS codes have been studied as well \cite{Jin01}. One major topic on LCD codes is to construct asymptotically good LCD codes. In \cite{Mass92}, Massey showed that  there exist asymptotically good LCD codes by establishing a relationship between LCD codes and linear codes. Meanwhile, he raised a question on whether LCD codes can achieve the Gilbert-Varshamov bound. Later, using the hull dimension spectra of linear codes, Sendrier showed that LCD codes can meet the asymptotic Gilbert-Varshamov bound \cite{Send04}. Until now, the GV bound still remains to be the best asymptotical lower bound.

Recently, Mesnager, Tang and Qi \cite{Mass92} showed that, under two conditions, algebraic geometry codes are equivalent to LCD codes. In general, these two conditions are not easily satisfied for an arbitrary curve. Thus, no asymptotical result is derived from their result. Instead, they presented a few examples of curves such as projective line, elliptic curve and Hermitian curves etc that satisfy their conditions.

In this paper, we show that LCD codes  exceed the Gilbert-Varshamov  bound by constructing a class of LCD codes that are equivalent to algebraic geometry codes. The ideal works as follows:  firstly we show that a linear code can be turned into an LCD code under two conditions; then we do some counting on algebraic geometry codes and show that these two conditions are satisfied for algebraic geometry codes if the underlying function fields
have many rational places and the code alphabet size is not too small; finally, we obtain LCD codes exceeding the asymptotic Gilbert-Varshamov bound by applying two function field towers.

The paper is organized as the followings. In Section II, we introduce some basic definitions and terminologies on LCD codes, function fields, algebraic geometry codes and some function field towers. In Section III, we first show that a linear code can be turned into an LCD code under two conditions. Then some counting on algebraic geometry codes is presented. Finally, we show that  LCD codes are equivalent to algebraic geometry codes and exceed the asymptotic Gilbert-Varshamov bound.

\section{Preliminaries }\label{sec:2}
In this section, we introduce some basic results on LCD codes, function fields, algebraic geometry codes and function field towers.

\subsection{Linear complementary dual}\label{subsec:2.1} For two vectors $\ba=(a_1,\dots,a_n), \bb=(b_1,\dots,b_n)$ in
$\F_q^n$, the Euclidean inner product is defined by
$\langle\ba,\bb\rangle=\sum_{i=1}^na_ib_i$. For a linear code $C$ over $\F_q$,
the { Euclidean dual} of $C$ is defined by
\[C^{\perp}:=\{\bv\in\F_q^n:\;\langle\bv,\bc\rangle=0, \; \forall \bc\in C\}.\]
A linear code is called linear complementary dual  if $C\cap C^{\perp}=\{\bo\}$.

For two vectors $\bu=(u_1,\dots,u_n),\bv=(v_1,\dots,v_n)\in\F_q^n$, denote by  $\bu*\bv$ the Schur product $(u_1v_1,\dots,u_nv_n)$. In particular, we denote by $\bv^2=\bv*\bv$.
For a vector $\ba=(a_1,\dots,a_n)\in(\F_q^*)^n$, denote by $\ba^{-1}$ the vector $(a_1^{-1},\dots,a_n^{-1})\in(\F_q^*)^n$. Furthermore, for a linear code $C$, we denote by $\ba*C$ the linear code $\{\ba*\bc:\; \bc\in C\}.$
It is clear that $C$ and $\ba*C$ are equivalent. Furthermore, it is easy to verify that $(\ba*C)^{\perp}=\ba^{-1}*C^{\perp}$.

\subsection{Function fields and algebraic geometry codes}\label{subsec:2.2}

To construct LCD codes that are equivalent to algebraic geometry codes, we need to
recall some basic definitions and results of algebraic function fields and
algebraic geometry codes. The reader may refer to \cite{Stich08, Ts
Vl} for the detail.

Let $F$ be a function field of genus $g$ defined over $\F_q$.
An element of $F$ is called a function. A place $P$ of $F$ is the maximal ideal in a valuation ring $O$. The residue field $O/P$ is isomorphic to  an extension field over $\F_{q}$. The degree of  $P$ is defined to be $[O/P:\F_q]$. A place of degree one is called rational. The normalized discrete valuation corresponding to a place $P$  is written as $\nu_P$. We use $\PP_F$ to denote the set of all places of $F$.

A divisor $G$ of $F$ is a formal sum $\sum_{P\in\PP_F}m_PP$ with only finitely many nonzero integers $m_P$. The support of $G$ is defined to be $\{P\in\PP_F:\; m_P\neq 0\}$.  The degree of $G$ is defined to be $\sum_{P\in\PP_F}m_P\deg(P)$. Divisor $G=\sum_{P\in\PP_F} m_PP$ is said to be bigger than or equal to divisor $D=\sum_{P\in\PP_F} n_PP$ if $m_P\ge n_P$ for all $P\in\PP_F$. A divisor $G=\sum_{P\in\PP_F}m_PP$ is said to be effective, denoted by $G\ge 0$ if $m_P\ge 0$ for all $P\in\PP_F$. For a nonzero function $f$, the principal divisor ${\rm div}(f)$ is defined to be $\sum_{P\in\PP_F} \nu_P(f)P$.

For a divisor $G$, the Riemann-Roch space associated to $G$ is
\[\mL(G)=\{f\in F\setminus\{0\}:\; {\rm div}(f)+G\geq0\}\cup\{0\}.\]
Then $\mL(G)$ is a finite-dimensional vector space over $\F_q$ and we denote its dimension by $\ell(G)$. By the Riemann-Roch theorem we have
\[\ell(G)\geq \deg(G)+1-g,\]
where the equality holds if $\deg(G)\geq 2g-1$.

Let $P_1,\dots,P_n$ be pairwise distinct rational places of $F$ and let $D=P_1+\dots+P_n$. Choose a
divisor $G$ of $F$ such that ${\rm supp}( G)\bigcap {\rm supp}( D)=\emptyset$. Then $\nu_{P_i}(f)\geq0$ for all $1\leq i\leq n$ for any $f\in \mL(G)$.

Consider the map
\[\Psi:\mL(G)\rightarrow \F_q^n,\quad f\mapsto(f(P_1),\dots,f(P_n)).\]
Obviously the image of  $\Psi$ is a $q$-ary linear code.  This
code is defined to be an
algebraic-geometry code (or AG code for short), denoted by $C_L(D,G)$. If $\deg(G)<n$, then
$\Psi$ is an embedding and we have $\dim(C_L(D,G))=\ell(G)$.
By the Riemann-Roch theorem we can estimate the parameters of an AG code
(see \cite[Theorem 2.2.2]{Stich08}).

\begin{lemma}\label{lem:2.1} $C_L(D,G)$ is an $[n,k,d]$-linear code over $\F_q$ with parameters
\[k=\ell(G)-\ell(G-D),\quad  d\geq n-\deg(G).\]
\begin{itemize}
\item[{\rm (a)}] If $G$ satisfies $g\leq \deg(G)<n$, then
\[k=\ell(G)\geq \deg(G)-g+1.\]
\item[{\rm (b)}] If additionally $2g-2<deg(G)<n$, then $k=\deg(G)-g+1$.
\end{itemize}
\end{lemma}

Now we discuss the Euclidean dual of the AG code
$C_L(D,G)$.

The differential space of $F$ is defined to be
\[\Omega_F:=\{fdx:f\in F\},\]
where $\nu_Q(x)$ is coprime with $q$ for some place $Q$.
This is a one-dimensional space over $F$. For a  place $P$ and a function $t$ with $\nu_P(t)=1$, we define $\nu_P(fdt)=\nu_P(f)$. The divisor associated with a nonzero differential $\omega$ is defined to be ${\rm div}(w)=\sum_{P\in\PP_F}\nu_P(\omega)P$. Such a divisor is called a canonical divisor. Every   canonical divisor has degree $2g-2$, where $g$ is the genus of $F$. Furthermore, any two canonical divisors are equivalent.  Now, if $P$ is a rational place and $\nu_P(fdt)\ge -1$,  we define the residue of $fdt$ at $P$ to be $(ft)(P)$, denoted by $\res_{P}(fdt)$.

 For a divisor $G$, we define the $\F_q$-vector space
\[\Omega(G)=\{\omega\in \Omega_F\setminus\{0\}:\;{\rm div}(\omega)\ge G\}\cup\{0\}\]
and denote the dimension of $\Omega(G)$ by $i(G)$. Then one has the following relationship
\[i(G)=\ell(K-G),\]
where $K$ is a canonical divisor.

 We define the code $C_\Omega(D,G)$ as
\[C_\Omega(D,G)=\{(\res_{P_1}(\omega),\dots,\res_{P_n}(\omega)):\;\omega\in
\Omega(G-D)\}.\]
We have the following results \cite[Theorem 2.2.7 and Proposition 2.2.10]{Stich08}
\begin{lemma}\label{lem:2.2}  $C_\Omega(D,G)$ is an $[n,k^{\perp},d^{\perp}]$-linear code over $\F_q$ with parameters
\[k^{\perp}=i(G-D)-i(G),\quad  d^{\perp}\geq \deg( G)-(2g-2).\]
\begin{itemize}
\item[{\rm (a)}] If $G$ satisfies $2g-2\leq \deg(G)<n$, then
\[k^{\perp}= n-\deg(G)+g-1.\]
\item[{\rm (b)}] There exists a nonzero differential $\eta\in\Omega_F$ such that $\nu_{P_i}(\eta)=-1$ for $i=1,\dots,n$ and $C_\Omega(D,G)=\bv*C_L(D,D-G+{\rm div}(\eta))$ for some $\bv\in(\F_q^*)^n$.
\end{itemize}
\end{lemma}
To distinguish two classes of algebraic geometry codes $C_L(D,G)$ and $C_\Omega(D,G)$, we call $C_L(D,G)$ a functional (AG) code and $C_\Omega(D,G)$ a differential (AG) code.

\subsection{Two function field towers}\label{subsec:2.3} In this subsection, we introduce two function field towers with many rational places. These function fields will be used to construct LCD algebraic geometry codes exceeding the Gilbert-Varshamov bound. The reader may refer to \cite{BBGS,GS} for the details on these two towers.

\noindent {\bf The first tower.} The first tower is defined over $\F_q$ with $q=r^2$ for a prime power $r$. The function field is $F=\F_q(x_1,x_2,\dots,x_t)$, where $x_i$ are transcendental elements over $\F_q$ satisfying the following recursive equations
\[x_{i+1}^r+x_{i+1}=\frac{x_i^r}{x_i^{r-1}+1}\]
for $i=1,2,\dots,t-1$. Let $N(F)$ and $g(F)$ denote the number  of rational places of $F$ and the genus of $F$, respectively. Then one has
\begin{equation}\label{eq:2.1}N(F)\ge r^{t-1}(q-1)+1\ge (\sqrt{q}-1)g(F)+1.\end{equation}
Furthermore, $g(F)\rightarrow\infty$ as $t\rightarrow\infty$.

\noindent {\bf The second tower.} The second tower is defined over $\F_q$ with $q=2^{2m+1}$ for an integer $m\ge 1$. For an integer $j$, denote by $\Tr_j(T)=T+T^2+\cdots+T^{2^{j-1}}$. The function field is $F=\F_q(x_1,x_2,\dots,x_t)$, where $x_i$ are transcendental elements over $\F_q$ satisfying  the following recursive equations
\[\Tr_m\left(\frac{x_{i+1}}{x_i^{2^{m+1}}}\right)+\Tr_{m+1}\left(\frac{x_{i+1}^{2^{m}}}{x_i}\right)=1\]
for $i=1,2,\dots,t-1$. Let $N(F)$ and $g(F)$ denote the number  of rational places of $F$ and the genus of $F$, respectively. Then one has $g(F)\rightarrow\infty$ as $t\rightarrow\infty$, and
\begin{equation}\label{eq:2.2}\lim_{t\rightarrow\infty}\frac{N(F)}{g(F)-1}\ge \frac{2(2^{m+1}-1)(2^m-1)}{3(2^m-1)+1}.\end{equation}

\subsection{Gilbert-Varshamov bound}\label{subsec:2.4}
The Gilbert-Varshamov bound is a benchmark for good codes. It has been shown that with a high probability, a  linear code achieves the Gilbert-Varshamov bound.
It was proved that LCD codes can attain the following asymptotical GV bound in \cite{Send04}.
\begin{lemma}[Asymptotical Gilbert-Varshamov bound]
For every $q$ and $\delta\in(0,1-1/q)$, there exists a family $\{C_i=[n_i,k_i,d_i]\}$ of $q$-ary LCD code such that $n_i\rightarrow\infty$ as $i\rightarrow\infty$, $R=\lim_{i\rightarrow\infty}\frac{k_i}{n_i}$ and $\Gd=\lim_{i\rightarrow\infty}\frac{d_i}{n_i}$ satisfy
\[R\ge 1-H_q(\delta),\]
where $H_q(x)=x\log_q(q-1)-x\log x-(1-x)\log_q(1-x)$ is the $q$-ary entropy function.
\end{lemma}
\section{Construction}\label{sec:3}
From now onwards, we assume that $\F_q$ has characteristic equal to $2$.  In this first subsection, we show that a linear code can be turned into an LCD code under two conditions. In the second subsection, we do some counting on algebraic geometry codes and show that these two conditions are satisfied for algebraic geometry codes if the underlying function fields have many rational points and $q$ is not too small. In the third subsection, we  apply the two function field towers in Subsection \ref{subsec:2.3} to obtain LCD algebraic geometry codes that exceed the asymptotic Gilbert-Varshamov bound.
\subsection{A general construction of LCD codes}\label{subsec:3.1} This subsection shows that under certain conditions, a linear code can be turned into an equivalent code which is LCD.
For a  $q$-ary $[n,k]$-linear code $C$ and a subset $I\subseteq\{1,2,\dots,n\}$, define the sets
\begin{equation}\label{eq:3.1}
S_I(C)=\{\bc\in C:\; \supp(\bc)=I\};\quad S_I(C^{\perp})=\{\bc\in C^{\perp}:\; \supp(\bc)=I\}
\end{equation}

Furthermore, for a codeword $\bu\in C$, define the set
\begin{equation}\label{eq:3.2}
T_{\bu}(C)=\{\bv\in (\F_q^*)^n:\; \bv^2*\bu\in C^{\perp}\}.
\end{equation}

\begin{lemma}\label{lem:3.1} Assume that $C$ is a $q$-ary $[n,k,d]$-linear code. If there exists a function $\Gg_q(n)$ such that
\begin{equation}\label{eq:3.3}
q^{-k}(q-1)^n-(n-d)\Gg_q(n)2^nq^{k-n}\left(q-1\right)^n-(n-d)2^n(\Gg_q(n))^2(q-1)^{n-d}>0\end{equation}
and
\begin{equation}\label{eq:3.4}
|S_I(C)|\le q^{-(n-k)}\left(q-1\right)^w+ \Gg_q(n);\quad |S_I(C^{\perp})|\le q^{-k}\left(q-1\right)^w+ \Gg_q(n)\end{equation}
for any  subset $I\subseteq\{1,2,\dots,n\}$ with $|I|=w\ge 1$, then  there exists a vector $\bv\in (\F_q^*)^n$ such that $\bv*C$ is an LCD code.
\end{lemma}
\begin{proof} Let $T_{\bu}(C)$ be defined in \eqref{eq:3.2}. Assume that the cardinality of the union $\cup_{\bu\in C\setminus\{\bo\}}T_{\bu}(C)$ is less than $(q-1)^n$, i.e.,
\begin{equation}\label{eq:3.5}
\left|\bigcup_{\bu\in C\setminus\{\bo\}}T_{\bu}(C)\right|<(q-1)^n,
\end{equation}
 then one can find a vector $\ba\in (\F_q^*)^n$ such that $\ba\not\in \cup_{\bu\in C\setminus\{\bo\}}T_{\bu}(C)$. This implies that
  $\ba^2*\bu\not\in C^{\perp}$, i.e., $\ba*\bu\not\in\ba^{-1}*C^{\perp}=(\ba*C)^{\perp}$ for all $\bu\in C\setminus\{\bo\}$. Hence, in this case,  the code $\ba*C$ is LCD. Thus, it is sufficient to show that the inequality \eqref{eq:3.5} holds.

 For a codeword $\bu\in C\setminus\{\bo\}$, let $I$ be the support of $\bu$. Consider the set
 \begin{equation}\label{eq:3.6}
R_{\bu}(C)=\{\bv^2*\bu:\; \bv\in T_{\bu}(C)\}.
\end{equation}
It is clear that $R_{\bu}(C)\subseteq S_I(C^{\perp})$. Furthermore, we have the relation $|T_{\bu}(C)|=(q-1)^{n-w}|R_{\bu}(C)|$, where $w=|I|$.
Thus, by \eqref{eq:3.4}, we have
 \begin{equation}\label{eq:3.7}
|T_{\bu}(C)|=(q-1)^{n-w}|R_{\bu}(C)|\le (q-1)^{n-w}|S_I(C^{\perp})|\le q^{-k}\left(q-1\right)^n+\Gg_q(n)(q-1)^{n-w}.
\end{equation}
Denote by $X$ the set $\cup_{\bu\in C\setminus\{\bo\}}T_{\bu}(C)$. Then
\begin{eqnarray*}|X|&\le& \sum_{\bu\in C\setminus\{\bo\}}|T_{\bu}(C)|
\le  \sum_{\bu\in C\setminus\{\bo\}}\left(q^{-k}\left(q-1\right)^{n}+\Gg_q(n)(q-1)^{n-\wt(\bu)}\right)\\
&=&q^{-k}\left(q-1\right)^n(q^k-1)+\sum_{w=d}^n\sum_{|I|=w}\Gg_q(n)|S_I(C)|(q-1)^{n-w}\\
&\le &(q-1)^n-q^{-k}(q-1)^n+\Gg_q(n)\sum_{w=d}^n{n\choose w}\left(q^{k-n}\left(q-1\right)^w+\Gg_q(n)\right)(q-1)^{n-w}\\
&\le &(q-1)^n-q^{-k}(q-1)^n+\Gg_q(n)2^n\sum_{w=d}^n\left(q^{k-n}\left(q-1\right)^n+\Gg_q(n)(q-1)^{n-w}\right)\\
&\le &(q-1)^n-q^{-k}(q-1)^n+(n-d)\Gg_q(n)2^nq^{k-n}\left(q-1\right)^n+(\Gg_q(n))^22^n\sum_{w=d}^n(q-1)^{n-w}\\
&\le &(q-1)^n-q^{-k}(q-1)^n+(n-d)\Gg_q(n)2^nq^{k-n}\left(q-1\right)^n+(n-d)2^n(\Gg_q(n))^2(q-1)^{n-d}\\
&<&(q-1)^n\quad \mbox{(by \eqref{eq:3.3})}.
\end{eqnarray*}
This completes the proof.
\end{proof}

\begin{rmk}{\rm When the characteristic of $\F_q$ is odd, then we have the relation $|T_{\bu}(C)|=2^w(q-1)^{n-w}|R_{\bu}(C)|$ in contrast with the first equality of  \eqref{eq:3.7} in the  proof of Lemma \ref{lem:3.1}. The extra factor $2^w$ destroys our inequality \eqref{eq:3.5}. However, in this case, $|R_{\bu}(C)|$ should be much smaller than $|S_I(C^{\perp})|$ since not every vectors in $S_I(C^{\perp})$ has the form $\bv^2*\bu$ for a fixed $\bu$. We are not sure if we can analyze the relation between $|R_{\bu}(C)|$ and $|S_I(C^{\perp})|$ properly so that we still have the inequality \eqref{eq:3.7} in this case.
}
\end{rmk}
\subsection{Counting on AG codes}\label{subsec:3.1}
The main purpose of this subsection is to show that algebraic geometry codes satisfy the conditions \eqref{eq:3.3} and \eqref{eq:3.4}.

\begin{lemma}\label{lem:3.2} Let $C$ be the algebraic geometry code $C_L(D, G)$ of length $n$ defined over a function field of genus $g$. Assume that $ g\le \Gl n$ with a constant $\lambda>0$. Put $\Gg_q(n)=2(2q^{\Gl })^n$. If $2g-1\le \deg(G)\le n-1$, then \[
|S_I(C)|\le q^{-(n-k)}\left(q-1\right)^w+\Gg_q(n);\quad |S_I(C^{\perp})|\le q^{-k}\left(q-1\right)^w+\Gg_q(n)\]
for any  subset $I\subseteq\{1,2,\dots,n\}$ with $|I|=w\ge 1$.
\end{lemma}
\begin{proof} Let us prove the inequality on $|S_I(C)|$ first. Without loss of generality, we may assume that $I=\{1,2,\dots,w\}$.
Let $m$ denote the degree of $G$.
For a codeword $(f({P_1}),\dots,f({P_n}))$ in $C_L(D, G)$ with $f\in
\mL\left(G\right)$, the $i$-th coordinator $f({P_i})$ is zero if and only if $f\in
\mL\left(G-P_i\right)$. Thus,
\begin{equation}\label{eq:3.8}S_I(C)=\left\{(f(P_1),\dots,f(P_n)):\;f\in\mL\left(G-\sum_{j=n-w+1}^nP_j\right)\setminus\bigcup_{i=1}^{w}\mL\left(G-\sum_{j=n-w+1}^nP_j-P_i\right)\right\}.\end{equation}
If $w<n-m$, then $S_I(C)= \emptyset$ and the desired result is clear. Now assume that $w\ge n-m$.
If $m-(n-w)\le 2g-2$, i.e., $w\le n-m+2g-2$, then $|S_I(C)|\le \left|\mL\left(G-\sum_{j=n-w+1}^nP_j\right)\right|\le q^g<\Gg_q(n)$. The desired result holds in this case as well.

Now assume that  $w\ge n-m+2g-1$.
 We denote   $\mL(G-\sum_{j=n-w+1}^nP_j-P_i)$ by $A_i$.  By the equation \eqref{eq:3.8} and the inclusion-exclusion principle, we have
  \begin{eqnarray*}
|S_I(C)|&=&\left|\mL\left(G-\sum_{j=n-w+1}^nP_j\right)\right|-\sum_{i=1}^{w}|A_i|+\sum_{h,k}|A_h\cap A_k|+ \dots+(-1)^{w}\sum_{i_1,\dots,i_w}\left|\bigcap_{j=1}^{w}A_{i_j}\right|\\
&=&\sum_{j=0}^{m-n+w-2g+1}(-1)^j{w\choose j}q^{m-n+w-g+1-j}+\sum_{j=m-n+w-2g+2}^{w}(-1)^j\sum_{i_1,\dots,i_{j}}\left|\bigcap_{r=1}^{j}A_{i_r}\right|\\
 &=&q^{m-g+1+w-n}\left(1-\frac1q\right)^w+c=q^{k-n}\left(q-1\right)^w+c,
\end{eqnarray*}
where
 \begin{eqnarray*}c&=&\sum_{j=m-n+w-2g+2}^{w}(-1)^j\sum_{i_1,\dots,i_{j}}\left|\bigcap_{r=1}^{j}A_{i_r}\right|-\sum_{j=m-n+w-2g+2}^{w}(-1)^j{w\choose j}q^{m-n+w-j-g+1}\\
 &\le&\sum_{j=m-n+w-2g+2}^{w}\sum_{i_1,\dots,i_{j}}\left|\bigcap_{r=1}^{j}A_{i_r}\right|+\sum_{j=m-n+w-2g+2}^{w}{w\choose j}q^{m-n+w-j-g+1}\\
 &=&q^{g}\sum_{j=0}^{w}{w\choose j}+q^{g-1}\sum_{j=0}^{w}{w\choose j}\le 2q^{g}\times 2^{w+1}\le q^{\Gl n}2^{n}=\Gg_q(n).\end{eqnarray*}
This proves the inequality on $S_I(C)$.

By Lemma \ref{lem:2.2}, there exists a vector $\bv\in(\F_q^*)^n$  and a nonzero differential $\eta$ such that $C^{\perp}=C_\Omega(D,G)=\bv*C_L(D,D-G+(\eta))$. The desired result on $S_I(C^{\perp})$ follows from the fact that $S_I(C^{\perp})=S_I(\bv*C_L(D,D-G+(\eta))=S_I(C_L(D,D-G+(\eta))$ since we have proved the result for the functional codes $C_L$.
\end{proof}
Lemma \ref{lem:3.2} shows that the inequalities \eqref{eq:3.4} are satisfied for algebraic geometry codes. Next, we are going to show that the inequality \eqref{eq:3.3} is also satisfied for algebraic geometry codes.

\begin{lemma}\label{lem:3.3} Let $C$ be the algebraic geometry code $C_L(D, G)$ of length $n$ defined over a function field of genus $g$ with $2g-1\le \deg(G)\le n-1$. Assume that $ g\le \Gl n$ with a constant $\lambda>0$. Put $\Gg_q(n)=2(2q^{\Gl })^n$. Assume that $\lim_{n\rightarrow\infty}\frac kn=R$ and $\lim_{n\rightarrow\infty}\frac dn=\Gd$. If
 \begin{equation}\label{eq:3.1a}
 (1-2R-\Gl)\log_2q-2>0,\quad \Gd\log_2(q-1)-R\log_2q-3-2\Gl>0,
 \end{equation}
 then for all sufficiently large $n$, one has \[q^{-k}(q-1)^n-(n-d)\Gg_q(n)2^nq^{k-n}\left(q-1\right)^n-(n-d)2^n(\Gg_q(n))^2(q-1)^{n-d}>0. \]
\end{lemma}
\begin{proof} By \eqref{eq:3.1a}, there exists $\Ge>0$ such that  \begin{equation}\label{eq:3.1b}(1-2R-\Gl)\log_2q-2\ge \Ge,\quad \Gd\log_2(q-1)-R\log_2q-3-2\Gl\ge \Ge.\end{equation}
Put
\[A=\frac{q^{-k}(q-1)^n}{(n-d)\Gg_q(n)2^nq^{k-n}\left(q-1\right)^n}=\frac{q^{-k}}{(n-d)\Gg_q(n)2^nq^{k-n}}.\]
Then we have
\begin{eqnarray*}
\lim_{n\rightarrow\infty}\frac{\log_2A}{n}&=&\lim_{n\rightarrow\infty}\left(-\frac kn\log_2q-\frac1n\log_2(n-d)-\frac{2n+1}n-\Gl\log_2q-\frac{k-n}n\log_2q\right)\\
&=&(1-2R-\Gl)\log_2q-2\ge \Ge.
\end{eqnarray*}
 Put
\[B=\frac{q^{-k}(q-1)^n}{(n-d)2^n(\Gg_q(n))^2(q-1)^{n-d}}=\frac{q^{-k}(q-1)^d}{(n-d)2^n(\Gg_q(n))^2}.\]
Then we have
\begin{eqnarray*}
\lim_{n\rightarrow\infty}\frac{\log_2B}{n}&=&\lim_{n\rightarrow\infty}\left(-\frac kn\log_2q+\frac dn\log_2(q-1)-\frac1n\log_2(n-d)-\frac{3n+2}n-2\Gl\log_2q\right)\\
&=&\Gd\log_2(q-1)-R\log_2q-3-2\Gl\ge \Ge.
\end{eqnarray*}
Thus, we have $A>2^{\Ge n/2}$ and $B>2^{\Ge n/2}$ for all sufficiently large $n$. This implies that $1/A<1/2$ and $1/B<1/2$ for sufficiently large $n$. Therefore, one has $1-1/A-1/B>0$ for all sufficiently large $n$, i.e.,
\[q^{-k}(q-1)^n-(n-d)\Gg_q(n)2^nq^{k-n}\left(q-1\right)^n-(n-d)2^n(\Gg_q(n))^2(q-1)^{n-d}>0. \]
The proof is completed.
\end{proof}
\subsection{LCD algebraic geometry codes}
Define
\[\Gl_q=\left\{\begin{array}{ll}
\frac1{2^r-1}&\mbox{if $q=2^{2r}$ for some integer $r\ge 1$}\\
&\\
\frac{3(2^r-1)+1}{2(2^{r+1}-1)(2^r-1)}&\mbox{if $q=2^{2r+1}$ for some integer $r\ge 1$}.
\end{array}\right.\]
By Subsection \ref{subsec:2.3}, there exists a function field family $\{F/\F_q\}$ such that $N(F)\ge g(F)/\Gl_q+1$ and $g(F)\rightarrow\infty$.
\begin{theorem}\label{thm:3.4} Let $q$ be a power of $2$.
If \begin{equation}\label{eq:3.10}2\Gl_q<R<\min\left\{\frac12\left(1-\Gl_q-\frac2{\log_2q}\right),\frac{\log_2(q-1)-({1+\log_2(q-1)})\Gl_q-3}{\log_2q(q-1)}\right\},\end{equation} or
\begin{equation}\label{eq:3.1c}1-\min\left\{\frac12\left(1-\Gl_q-\frac2{\log_2q}\right),\frac{\log_2(q-1)-({1+\log_2(q-1)})\Gl_q-3}{\log_2q(q-1)}\right\}<R<1-2\Gl_q,\end{equation}
then there exist LCD codes that are equivalent to algebraic geometry codes with rate $R$ and relative minimum distance $\Gd$ achieving the Tsfasman-Vludat bound asymptotically, i.e.,
\begin{equation}\label{eq:3.1d}R+\Gd\ge 1-\Gl_q.\end{equation}
\end{theorem}
\begin{proof} If the desired result is true under the inequalities \eqref{eq:3.10}, then by considering dual codes, the desired result is also true under the inequalities \eqref{eq:3.1c}. Thus, we prove the result by only assuming the inequalities \eqref{eq:3.10}.

Let $\{F/\F_q\}$ be a function field family such that $N(F)\ge g(F)/\Gl_q+1$ and $g(F)\rightarrow\infty$. Denote $g(F)$ simply by $g$. Put $n=N(F)-1$ and choose $n+1$ distinct rational places $P_1,P_2,\dots,P_n,P_{\infty}$. Put $D=\sum_{i=1}^nP_i$ and $G=mP_{\infty}$ with $m=\lfloor Rn\rfloor$. Then $2g-1\le m\le n-1$. Consider the algebraic geometry code $C=C_L(D,G)$ and its dual $C^{\perp}=C_\Omega(D,G)$. Put $\Gd=1-R-\Gl_q$.  It is clear that  the code $C$ has rate $R$ and relative minimum distance at least $\Gd$. Thus, by Lemma \ref{lem:3.1}, it is sufficient to show  that $C$ satisfies the inequalities \eqref{eq:3.3} and \eqref{eq:3.4} for all sufficiently large $n$.

By Lemma \ref{lem:3.2}, the inequality \eqref{eq:3.4} is satisfied for $\Gl=\Gl_q$. Now by  Lemma \ref{lem:3.3}, to show the inequality \eqref{eq:3.3} for all sufficiently large $n$, it is sufficient to show that the inequalities \eqref{eq:3.1a} are satisfied for all sufficiently large $n$.

It is easy to verify that the second inequality of \eqref{eq:3.10} is equivalent to the  inequalities of \eqref{eq:3.1a} for $\Gl=\Gl_q$. This completes the proof.
\end{proof}

\begin{cor} When $q\ge 128$ is a  power of $2$, then there exist LCD codes that are equivalent to algebraic geometry codes and exceed the asymptotic Gilbert-Varshamov bound in two intervals of $(0,1)$.
\end{cor}
\begin{proof} It is well known that \eqref{eq:3.1d} exceeds the asymptotic Gilbert-Varshamov bound in an interval $\Gd\in (a,b)\subset (0,1)$ for  $q\ge 128$. Now it is straightforward to verify that there are two subintervals $(a_1,b_1)$ and $(a_1,b_1)$ of $(a,b)$ such that, for $q\ge 128$, the rate $R=1-\Gd-\Gl_q$ lies in the range of \eqref{eq:3.10} for $\Gd\in (a_1,b_1)$, and in the range of \eqref{eq:3.1c} for $\Gd\in (a_2,b_2)$, respectively. This completes the proof.
\end{proof}

\end{document}